\documentclass[11pt]{article}
\usepackage{graphicx}
\usepackage{color}
\usepackage{ifthen}
\usepackage{microtype}
\usepackage{psfrag,float}
\usepackage{boxedminipage}
\usepackage{graphicx}
\usepackage{amsmath}
\usepackage{amsfonts}
\usepackage{amssymb}
\usepackage{authblk}
\usepackage{hyperref}
\usepackage[english]{babel}
\usepackage{blindtext}
\usepackage{amssymb,amsfonts,epsfig}
\usepackage[ruled,vlined]{algorithm2e}
    \newtheorem{theorem}{Theorem}
    \newtheorem{lemma}[theorem]{Lemma}
    \newtheorem{coro}[theorem]{Corollary}
\newenvironment{proof}{\par\noindent{\bf Proof:}}{\mbox{}\hfill$\Box$\\}

\newcommand{\ignore}[1]{}

\newcommand{\ric}{{\it RIC }}

\newcommand{\E}{{\mathbb E}}
\newcommand{\vecx}{\bar{X}}
\newcommand{\ignoretext}[1]{{}}

\newcommand{\commented}{no}

\ifthenelse{\equal{\commented}{no}}{
\newcommand{\jnote}[1]{}
\newcommand{\snote}[1]{}
\newcommand{\rnote}[1]{}
}

\begin{document}
\title{On tail estimates for Randomized Incremental Construction}
\author{ Sandeep Sen\\
   Department of CSE,\\
   I.I.T. Delhi, India\\
   {\small\texttt{ssen@cse.iitd.ac.in}}
}
\maketitle
\begin{abstract}
By combining several interesting applications of random sampling in geometric 
algorithms
like point location, linear programming, segment intersections, binary 
space partitioning,
Clarkson and Shor \cite{CS89} developed a general framework of randomized
incremental construction (\ric ). 
The basic idea is to add objects
in a random order and show that this approach yields efficient/optimal
bounds on {\bf expected} running time. Even quicksort can be viewed as a special
case of this paradigm. However, unlike quicksort, for most of these problems,
attempts to obtain sharper tail estimates on the running time had proved
inconclusive. Barring some results by \cite{MSW93,CMS92,Seidel91a}, the general question
remains unresolved.

In this paper we present some general techniques to obtain tail estimates for 
 \ric and and provide applications to some fundamental problems like Delaunay triangulations 
and construction of Visibility maps of intersecting line segments. The main result
of the paper centers around a new and careful application of Freedman's \cite{Fre75} 
inequality for Martingale concentration that overcomes the bottleneck of the better 
known Azuma-Hoeffding inequality. Further, 
we show instances where an \ric based algorithm 
may not have inverse polynomial tail estimates. In particular, we show that
the \ric time bounds for trapezoidal map 
can encounter a running time of $\Omega (n 
\log n\log\log n )$ with probability exceeding $\frac{1}{\sqrt{n}}$. This rules 
out inverse
polynomial concentration bounds around the expected running time.
\end{abstract}
\section{Introduction}
One of the most natural and elegant paradigm for designing geometric algorithms 
is randomized incremental construction or \ric for short. It can be viewed as 
generalization of Quicksort and evolved over a sequence of papers 
\cite{Mu88,Cl89}
eventually culminating in a very general framework of {\it configuration 
space} by Clarkson and Shor \cite{CS89}. The basic procedure is described in
Figure \ref{fig:ric0}.
\begin{figure}[t]
\begin{procedure}[H]%
  \nl ${\cal N} = [ x_1 , x_2 \ldots x_n ]$ : a random permutation of $S$. \;
 \nl $T \leftarrow \phi \ \ , H $ is the data structure properly initialized\;
 \nl \For{ $i = 1$ to $n$}{
  \nl $T \leftarrow T \cup  \{ x_i \} $\;
  \nl Update $H(T)$ }
\nl Return $ H(T)$ \;
\caption{RIC ($S $ )}
\end{procedure}
\caption{Randomized Incremental Construction }
\label{fig:ric0}
\end{figure}
Quicksort itself can be viewed through this paradigm as refinement of the
current partially ordered set (partitions) 
 by {\it inserting} the next splitter and updating the partitions. 
Some of the uninserted elements are further partitioned
because of the latest insertion. Although the worst case deterministic behavior can
be quite bad, the expected performance for a random insertion sequence 
(where are permutations are equally likely)
is quite efficient, and often optimal. 

A related but a somewhat distinct approach was developed in the work of
Seidel \cite{Seidel91,Seidel91a,Chew88} that maintains a solution inductively 
that is recomputed from scratch when the solution does not hold for the current
insertion. The closest pair can also be computed in a similar manner (\cite{KM:95}).
Although our techniques can be applied to the latter work also, we will focus 
primarily on the Clarkson-Shor paradigm of a configuration space.

An abstract configuration space, that we will refer to as $\Pi (S)$ is defined by
the given set $S$ of $n$ elements. A configuration $\sigma$ is defined
by $O(1)$ objects of $S$ that we will denote by $d(\sigma )$\footnote{We
adopt some of the notations from \cite{Mu94} and for completeness, we include some 
formal definitions in the Appendix.}. A configuration $\sigma$ is
a subset of the Euclidean space and $\ell (\sigma) = \sigma \cap 
\{ S - d(\sigma )\}$, i.e. the objects that intersect $\sigma$ not including 
$d(\sigma )$. $\Pi^{i} (S) = \{ \sigma : |\ell (\sigma )| = i \}$ and
$\Pi (S) = \bigcup_{i=0}^{n} \Pi^i (S)$. For analyzing \ric , $\Pi^0 (R)$ 
where $R \subset S$ is a randomly chosen subset often
turns out to be very important, that captures the {\it uninserted} elements
of $S$ and how they interact with the current partially constructed 
structure, denoted by $H(R)$. For notational simplicity, for $\sigma \in \Pi^{0} (R)$, 
$\ell (\sigma) = \sigma \cap S$ (instead of $\sigma \cap R$) which will be an important
parameter in the analysis. 
The reader is referred to \cite{CS89,Mu94} for further details
regarding this framework.   

When the next randomly chosen element $s \in S - R$ is included in $R$, $H(R)$
is updated and the cost of this contributes to the running time of \ric . In
\cite{CS89} , the data-structure is maintained as a {\it conflict graph} that
maintains relation between $\sigma \in \Pi^0 (R) $ and $S-R$ as a bipartite
graph. Clearly configurations are created and destroyed but the amortized cost
can be shown to be the cost of new configurations created and the ones destroyed
can be charged to the cost of past creation. Although the initial analysis in
\cite{CS89} was somewhat intricate and complex, subsequent papers 
\cite{Chew88,Seidel91}
simplified the analysis using a clever technique called {\it backward analysis}.
In this paper, we will appeal to the simpler analysis. Often the full
conflict graph information can be replaced by simpler relations (see 
\cite{GKS92,Seidel91}. However, the conflict graph approach is very general and
works for diverse problems.

For analyzing \ric based algorithms, Clarkson and Shor \cite{CS89} derived
many useful bounds based on properties of uniform random sampling that
generalized the results of Haussler and Welzl \cite{HW86} that essentially
gave a bound on $\max_{\sigma \in \Pi^0 (R)} | \ell (\sigma ) |$. We will exploit
such properties in the present paper - for proofs the reader can consult 
\cite{CS89,Mu94}. Henceforth $\ell (\sigma )$ will also be used as a notation for
$| \ell (\sigma )|$. 

While the primary focus was on deriving bounds on the expected running time 
of \ric , it
was felt that obtaining concentration bounds on the expected running time would
make the \ric more powerful and attractive. The conjecture is that the running
times are concentrated around its expected value but to the best of our 
knowledge, there has been little progress in this direction barring some papers
related to computing line segment intersections using \ric \cite{CMS92,MSW93} 
and on fixed dimensional linear programming \cite{Seidel91a}. 
It is also known that for problems like planar hulls, high probability 
bounds can be proved based on linear ordering that do not extend to higher
dimensions. 
Of course, by
resampling $\Omega ( \log n)$ times, we can obtain inverse polynomial 
concentration bounds at the expense of the increasing the running time by an
$O(\log n)$ factor.

In this paper, we revisit the problem and present some general methods to obtain
tail bounds for specific problems like {\it Delaunay triangulation, 3-D convex 
hulls, and line segment intersections} that are based on \ric. We obtain tail
estimates of the form $2^{-\alpha}$ for an $\alpha$ factor deviation from the
expected running times. This is similar to the bounds for resampling but it 
doesn't involve independent restarts of the algorithm.  
For the case of finding intersection of line segments, our bounds are not only 
better than \cite{MSW93} but also distinctly less involved in terms of calculations.

We also establish the tightness of such tail estimates by demonstrating a case
of trapezoidal maps (based on maintaining conflict lists) for which inverse
polynomial tail estimates can be obtained only for running times 
$\Omega (n \log n \log\log n)$ and rules out concentration bounds within constant factor
of expectation.

For analysis, we use Martingale inequalities based on the method of {\it
bounded variance} (as opposed to bounded difference) and the basic
martingale set up follows that of \cite{MSW93}. 

{\bf Remark} In \ric based algorithms, the term {\it running time} is often 
interchangeably used with {\it structural changes} caused by each insertion. The
data structures are consciously 
kept minimal and simple (like lists) that enable the running
times to be proportional to the structural changes that are explicitly 
handled. In this paper, we will use
the term {\it work} to denote structural changes and we will not attempt to 
analyze the precise running times.  
\subsection{Main techniques and organization}
We begin by introducing a useful probabilistic inequality, viz., Freedman's 
inequality \cite{Fre75} for Martingales that will be used to
model the running time of the generic \ric algorithms. 
In the following section, we illustrate the
use of this analysis technique on the classical algorithm quicksort that can be also viewed
through the lens of \ric . The application to quicksort doesn't yield any better
result but is a stepping stone to the more complex and general framework. 
In particular, even the more commonly used Azuma-Hoeffding bound is not known to be
useful for quicksort concentration bounds because its dependence on the 
worst-case bound (sum of bounded differences) making it ineffective.   
 
It is unlikely that the previously effective
techniques for concentration bound of quicksort extend to generic \ric 
because the intermediate structures in \ric are more complex and can be bounded
only in an {\it expected} sense. Note that in sorting, the intermediate structures
can be defined by exactly $i$ intervals after $i$ pivots are introduced. In the generic
\ric the intermediate data-structures may become much larger which explains why 
similar concentration bounds are hard to obtain. The Freedman's inequality is more
effective since it uses variance (expectation of the second moment) for which better
bounds can be obtained for the $i$-th step compared to the worst case.  

Starting with quicksort in section \ref{sec:qsort} we tackle increasingly complex scenarios
of \ric which can viewed as weaker bounds on the intermediate structures for which
we are trying to obtain concentration bounds. In the case of Delaunay triangulation,
in section \ref{sec:delaunay} the number of triangles in the $i$-th step is fixed but 
the number of new triangles created in the $i$-th step can be bound only 
in expectation.  In section \ref{sec:segment} we consider the 
case of line segment intersections where even the intermediate structure
can be bound only in an expected sense.

In the last section, we give concrete examples
of \ric to show that inverse polynomial concentration bounds are not feasible without changing 
some basic structure of the algorithm.

\section{Basic Tools}

\newcommand{\F}{\mathbb{F}}
\newcommand{\invsigma}{\pi^{-1}}

Let $S = \{ x_1 , x_2 \ldots x_n \}$ be a set of $n$ objects. A permutation $\pi$ of $S$ is 
a 1-1 function $\pi (i) = j$ where $i, j \in \{ 1, 2, \ldots n \}$ that produces a 
permutation $ x_{\pi (1)} , x_{\pi (2)} \ldots x_{\pi (n)} $. A {\it random} permutation
of $S$ is one of the $n!$ permutation function chosen uniformly at random. A $k$ {\it prefix} of a 
permutation $\pi$ is the sequence of the first $k$ objects and denoted by $\pi^{(k)}$
consisting of $x_{\invsigma (1)} , x_{\invsigma (2)} \ldots x_{\invsigma (k)} $. Note that the
permutation $x_3 , x_1 , x_2$ is defined as $\pi (1) = 2 ; \pi (2) = 3 ; \pi (3) = 1$, so
the permutation is $x_{\pi^{-1} (1)} , x_{\pi^{-1} (2)} , x_{\pi^{-1} (3)}$.  

Let $X_1 , X_2 \ldots 
X_{n} \ X_i \neq \{ X_1 , X_2 \ldots X_{i-1}\} $ where $X_i = x_{\invsigma (i)}$ 
corresponding to the random permutation $\pi$. 
Further, let ${\bar{X}}^{(k)}$ to denote a sequence of $k$ random variables.

Let $( \Omega , {\cal U} )$ denote the space of all possible permutations of $n$ objects and 
${\cal U}$ is the uniform probability distribution. 
For $0 \leq i \leq n$, let $\F_i$ consist of all permutations 
with fixed prefixes of
length $i$ and set $\F_0 = \epsilon$ (empty prefix). 
Then, $\F_i$ contains $\frac{n!}{(n-i)!}$ {\it blocks} corresponding to each of the
$i$ length prefixes.  
For example, if the set of objects is 
$\{ x_1 , x_2 , x_3 \}$, then the collection of events are as follows.
\begin{quote}
$\F_0 : \{  ( x_1  x_2 x_3 , x_1 x_3 x_1 , x_2 x_1 x_3 ,
 x_2 x_3 x_1 , x_3 x_1 x_2 , x_3 x_2 x_1 ) \}$ \\
$\F_1: \{  ( x_1 x_2 x_3 , x_1 x_3 x_2 ) , 
 ( x_2 x_1 x_3 , x_2 x_3 x_1 )  , ( 
x_3 x_1 x_2 , x_3 x_2 x_1 ) \}$ \\
$\F_2 : \{ ( x_1  x_2  x_3 ) , ( x_1  x_3  x_2 ) , ( x_2  x_1  x_3 ), ( x_2  x_3  
x_1 ), ( x_3  x_1  x_2 ) , ( x_3  x_2  x_1 ) ]$    
\end{quote}
The blocks within each $\F_i$ are indicated by $( \ )$.
More precisely $\F_i$ is a 
sigma algebra of the corresponding events that can be enumerated explicitly,
but omitted for brevity \footnote{ For example if $A, B \in \F$ then $A \cup B 
\in \F$ etc.} 
It can be easily verified that $\F_{i+1}$ is a refinement of $\F_i$ where each block of $\F_{i}$ 
is partitioned into $n-i$ subpartitions of $\F_{i+1}$.
These nested subcollections of $2^{\Omega} $ define a {\it filter} denoted by 
$\F_0 \subseteq \F_1 \subseteq
\ldots \F_{i-1} \subseteq \F_i \ldots \subseteq \F_n$
that can be used to define a sequence
of random variables $Y_i$ where $Y_i$'s are functions of $\F_i$'s. In particular, if the $Y_i$'s 
are $\F_i$ measurable, and
$\E [ Y_{i+1} | \F_i ] = Y_i $, then $Y_i$ is a {\it martingale} sequence \cite{fel68,fel71}.

In this context, let us define $Y_i = \E [ Y | {\bar{X}}^{(i)} ]$ for any well-defined 
random variable Y over the probability space $( \Omega , {\cal U} )$ where the conditioning is over
the events in $\F_i$. Then, it can be verified that
\[ \E [ Y_i ] = \E [ \E [ Y | {\bar{X}}^{(i)} ]] = \E [ \E [ Y | {\bar{X}}^{(i-1)}| X_i ]] = 
\E [ Y | {\bar{X}}^{(i-1)} ] = Y_{i-1} \]
since $\F_{i}$ is a refinement of $\F_{i-1}$.
The sequence $Y_i$ defines
a martingale sequence and is widely known as a {\it Doob Martingale} \cite{Doob40}.
It is more intuitive to visualize the above {\it filter} as a tree, 
where the level $i$ nodes correspond to {\it blocks} of $\F_i$ with arity $n-i$ and each 
sub-block is connected to its parent block by an edge directed from the parent.
Any node in the $j$-th level of this tree can be labelled by the (unique) sequence $X^{(j)}$  
leading to it.

In the context of analyzing \ric , it may be useful to club prefixes that result in
the same subset. For example, the length two prefixes of 
$[ x_1 , x_2 , x_3 ] , [x_2 , x_1 , x_3 ]$ lead to the subset $\{ x_1 , x_2 \}$. 
Similar to the previous example,
the {\it blocks} corresponding to distinct subsets can be enumerated as
\begin{quote}
$\F'_0 : \{  ( x_1  x_2 x_3 , x_1 x_3 x_1 , x_2 x_1 x_3 ,
 x_2 x_3 x_1 , x_3 x_1 x_2 , x_3 x_2 x_1 ) ]$ \\
$\F'_1: \{  ( x_1 x_2 x_3 , x_1 x_3 x_2 ) : \{ x_1 \},
 ( x_2 x_1 x_3 , x_2 x_3 x_1 ) :\{ x_2 \} , ( x_3 x_1 x_2 , x_3 x_2 x_1 ): \{ x_3 \} ]$ \\
$\F'_2 : \{ ( x_1  x_2  x_3 , x_2  x_1  x_3 ) :\{ x_1 , x_2 \}, 
( x_1  x_3  x_2 , x_3  x_1  x_2 ): \{ x_1 , x_3 \} , 
( x_2  x_3  x_1 ,  x_3  x_2  x_1 ) : \{ x_2 , x_3 \} \}$\\
$\F'_3 : \{  ( x_1  x_2 x_3 , x_1 x_3 x_1 , x_2 x_1 x_3 ,
 x_2 x_3 x_1 , x_3 x_1 x_2 , x_3 x_2 x_1 ) : \{ x_1 , x_2 , x_3 \}\}$
\end{quote}
The subsets corresponding to the blocks are indicated with curly brackets. 
The reader may notice that these blocks do not generate a nested sequence of 
sigma algebra - for instance the subset $\{ x_1 , x_2 \}$ defined by the 
block $( x_1 x_2 x_3 , x_1 x_3 x_2 )$ in $\F'_2$ 
come from
different blocks of $\F'_1$ depending on whether $x_1$ or $x_2$ occurs before. 
Therefore, this collection of prefixes on the insertion sequence 
is not consistent with a martingale on collection of subsets. 

Instead, let us 
interpret the sequences as {\it deletion} sequence, viz., starting from $\{ x_1, x_2 , x_3 \}$, we will delete the elements according to a random permutation, finally leading to the empty 
say $\phi$. Then $\F'_1$ denotes the subsets corresponding to deletion of $x_1 , x_2 , x_3$
respectively. This way, the blocks (of subsets) form a nested sequence. As we will see, this
interpretation leads to running the \ric in the {\it reverse} direction with each step known
in the literature as {\it backward analysis} which simplifies analysis considerably in many
situations \cite{Seidel93}.

  
Let random variables $X_1 , X_2 \ldots $, denote the
successive random choices in this tree starting from the root 
where the first $i$ choices correspond to the blocks of $\F_i$ for $0 \leq i \leq n$. 
An edge is labelled by the (random) choice made at that level 
and also has an associated weight $w( \vecx^{(i-1)}, \vecx^{(i)} )$ that corresponds to the cost of the 
$i$-th incremental step. We will use $W()$ to denote an upper bound of $w()$ in the context 
of specific algorithms.
 Let $Y = \sum_{j=0}^{j=n} w( \bar{X}^{(j-1)} , \vecx^{(j)})$ be a random variable that corresponds to 
the sum
of the cost of the edges on a path that corresponds to the cost of the RIC. 
Let $ Y_i  = \E [ Y | \F_i ] = \E [ Y | \bar{X}^{(i)} ]$. 

As we have noted before, 
$Y_i$ is a {\it Doob's martingale} based on the
random variables $X_i $ and $Y_0$ denotes the expected running
time of the RIC. We would like to bound the deviation $| Y_n - Y_0 |$ for
any run of the algorithm with high probability, which in the context of this paper will 
be inverse polynomial, unless otherwise mentioned.
Likewise the random deletion sequence also defines a {\it Doob's martingale} on the subsets
which will be referred to as the {\it backward-sequence} martingale (BSM henceforth).

The following martingale tail bound is the basis of many later results in this 
paper which is distinct from Azuma's inequality and referred to as the 
{\it Method of bounded variance}.

\begin{theorem}[Freedman\cite{Fre75}]
Let $X_1 , X_2 \ldots X_n$ be a sequence of random variables and
let $Y_k$, a function of $X_1 \ldots X_k$ be a martingale sequence, i.e.,
$\E[ Y_k | X_1 \ldots X_k ] = Y_{k-1}$ such that $\max_{1 \leq k \leq n} \{ | Y_k - Y_{k-1}|\} \leq 
M_n$.
Let
\[ W_k = \sum_{j=1}^{k} \E[ {( Y_j - Y_{j-1})}^2 | X_1 \ldots X_{j-1} ] = 
\sum_{j=1}^k Var ( Y_j | X_1 \ldots X_{j-1} ) \]
where $Var$ is the variance using $\E[Y_j] = Y_{j-1}$.
Then for all $\lambda$ and $W_n \leq \Delta^2 , \ \ \Delta^2 > 0 $ ,
\[ \Pr [ | Y_n - Y_0 | \geq \lambda ] \leq
2 \exp \left( - \frac{\lambda^2}{2 ( \Delta^2 + M_n \cdot \lambda /3)} \right) \]
\label{freedman}
\end{theorem}

Note that the term $\Delta^2$ can be bounded by $\sum_{j=1}^{n} \max_{X_1, X_2 \ldots X_j} 
Var ( Y_j | X_1 \ldots X_{j-1} )$ 
i.e., the worst case bounds over all choices
of length $j$ prefix $X^{(j)}$. If the inner term can be bounded by some function of $j$, say, 
$\omega (j)$, then 
we may obtain an upper bound on the probability of deviation 
for any sequence $\bar{X}^{(n)}$ as $\sum_{j=1}^{n} \omega (j)$ which can be viewed as a function of $n$. 

Further, we will actually use a minor variation of this result (see \cite{DP09}). 
Suppose $\Pr [ M_n \geq g(n) ] \leq \frac{1}{f(n)}$ for some non-decreasing functions
$g,f$. Then the overall bound becomes
\[ \Pr [ | Y_n - Y_0 | \geq \lambda ] \leq
2 \exp \left( - \frac{\lambda^2}{2 ( \Delta^2 + g(n) \lambda /3)} \right)  +
\frac{1}{f(n)} \]

Similarly it can also be extended to the case where $W_n \leq \Delta^2$ holds 
with probability $1 - \frac{1}{f(n)}$.
Henceforth, in the remaining paper, we will appeal to this version of Freedman's
inequality where the bounds on $M_n$ and $W_n$ hold with high probability. 
Often the term $\frac{1}{f(n)}$ will be the dominant term, so the
final tail bound will effectively be $O(\frac{1}{f(n)})$.

\section{Application to Quicksort and related problems}
\label{sec:qsort}

Let us consider quicksort in the \ric framework and without loss of generality,
let the input elements be $\{ 1, 2 \ldots n \}$. The $j$-th pivot, $1 \leq 
j \leq n$, partitions the input into $j+1$ ordered sets, by splitting some
existing partition $P$. Any element $x \in P$ is charged the cost of comparison 
with the pivot - any element $x' \not\in P$ is not charged. The running time of
the algorithm can be bounded by the cumulative charges accrued by each element.  
In this analysis we will bound the charge of each element {\it with high probability} 
(w.h.p.)\footnote{The acronym w.h.p. will
be used to denote  probability exceeding $1 - 1/n^{\alpha}$ for some appropriate 
constant $\alpha > 0$} and the 
overall running time bound follows from multiplying by $n$.
 
The associated weight with each edge
is either 1 or 0 
depending on whether the latest random choice is one of the 
boundary elements of the interval containing $x$. 
We define a random variable 
\begin{equation}
 I^{x}_{j} = \begin{cases}
1 \mbox{ if interval containing
$x$ changes in step $j$} \\
0 \mbox{  otherwise}
\end{cases}
\label{probone}
\end{equation}
From backward analysis, the probability of this is
at most $\frac{2}{j}$ for a uniformly chosen child node \footnote{Using a 
simple trick by considering a circular ordering (see \cite{Seidel93}), this 
probability can be made exactly equal to $\frac{2}{j}$.}. 
For completeness, we have included a detailed description of backward analysis
in the appendix. 
We will also omit the superscript $x$
and just use $I_j$ since we will obtain a worst case bound over all choices of $x$.
The reader may note that the bound on $\E [ I_j ] $ is only a function of $j$ and not
$\bar{X}^{(j)}$ over all random choices of {\it any} prefix of $j$ elements. 

It will also help to focus on the BSM for quicksort. A random deletion sequence creates 
a nested sequence of random subsets starting from the all the elements and ending in the
empty sequence. An edge of this tree $(K, K - \{ y \})$ is given a value 1 for a subset $K$ 
and an element $y \in K$ if in the (forward) quicksort algorithm, selecting $y$ as a pivot and
leading to $K$ (all the pivots selected) forces a comparison between $y$ and $x$. Clearly
two edges from any subset will be given a value 1, so that the expected cost for a random
deletion is $\frac{2}{n-j}$ in the $j$-th level, $n \geq j \geq 0$. 
Figure \ref{qsortfig} gives a depiction of this random variable in  
the quicksort process. 

Consider a path ${\cal P} = v_0 v_1 \ldots v_n$ from root to a leaf-node in this tree. 
The cost of this path is given by $w({\cal P}) = \sum_{i=1}^{n} w( v_i , v_{i+1})$. A
{\it random} path corresponds to one where $v_{i+1}$ is a child of $v_i$ chosen uniformly
at random among the $n-i$ children. The expected cost of such a random path is given by
 \[ \E_{random\ {\cal P}}[ w( {\cal P}) ]  =  \E[ \sum_{i=1}^n w( V_i , V_{i+1}) ] \mbox{ where }
V_{i+1} \mbox{ is a random child of node } V_i \]
Let $\E_{j} [Z]$ denote $\E[ Z | \vecx^{(j)}]$ for some random variable
$Z$. Note that $\vecx^{(j)}$ represents a fixed path from root to level $j$ of this tree 
corresponding to the deletion sequence $X_1 X_2 \ldots X_j$, say node $V_j$. 
Then,
\[ \E_{j} [ Y] = Y_j = \sum_{k=0}^{j-1} w( X_k , X_{k+1}) + \E [\sum_{k=j}^{n-1} 
 w( V_k , V_{k+1} ] = \sum_{k=0}^{j-1} w( X_k , X_{k+1}) + \sum_{k=j+1}^{n}
 \E [ I_k ] \] 
It follows that $Y_0 = 2 H_n$ and we want to obtain a tail estimate for $Y_n - Y_0$.

We can compute 
\begin{align*}
 Y_j - Y_{j-1} & = w( X_{j-1}, X_j ) + \left( \sum_{k=j+1}^{n}
 \E [ I'_k ] \right) - \left( \sum_{k=j}^{n} \E [ I_k ]\right)  \\ 
 \ & = I_j  - E[ I_j ] \mbox{ assuming $I_j , I_j'$'s have the same distribution} 
\end{align*}
 So $\E_{j-1} [ {( Y_j - Y_{j-1})}^2 ]  = \E [ {( I_j  - E[ I_j ])}^2]$ 
This shows that the value of $Y_j$ differs from $Y_{j-1}$ 
because of the specific choice random variable $X_j$. 
The above bound can be extended to a more 
general situations of \ric but where a single change can 
affect multiple "intervals" (more precisely, configurations).
 More specifically, for $W()$ not bounded by a constant we have the 
following generalization as long as $W_j$s have the same distribution across all nodes in
level $j$ for a random choice of the next node.
\begin{equation}
 \E_{X_j} [ {( Y_j - Y_{j-1} )}^2 ] 
 \leq  \E_{X_j} [ W_j^2 ] - \E^2 [ W_{j}] \leq \E_{X_j} [ W_j^2] 
\label{varbound}
\end{equation}
In the case of quicksort, we can complete the analysis as follows.
\begin{align*}
\E [ {( I_j  - E[ I_j ])}^2] &  = \E [ I_j^2 ] - \E^2 [ I_j ] \\
  \ & \leq \E [ I_j^2 ]- \frac{4}{{(n-j)}^2} \\
  \ & = \E [I_j^2 ] - \frac{4}{{(n-j)}^2} \\
  \ & \leq \frac{2}{n-j} \text{ since $I_j^2$ is also a 0-1 indicator rv } \\ 
\end{align*}
\begin{figure}[t]
\psfrag{yj}{{\small $Y_j $}}
\psfrag{yj1}{{\small $Y_{j-1} $}}
\psfrag{y0}{{\small $Y_0 $}}
\includegraphics[width=3in]{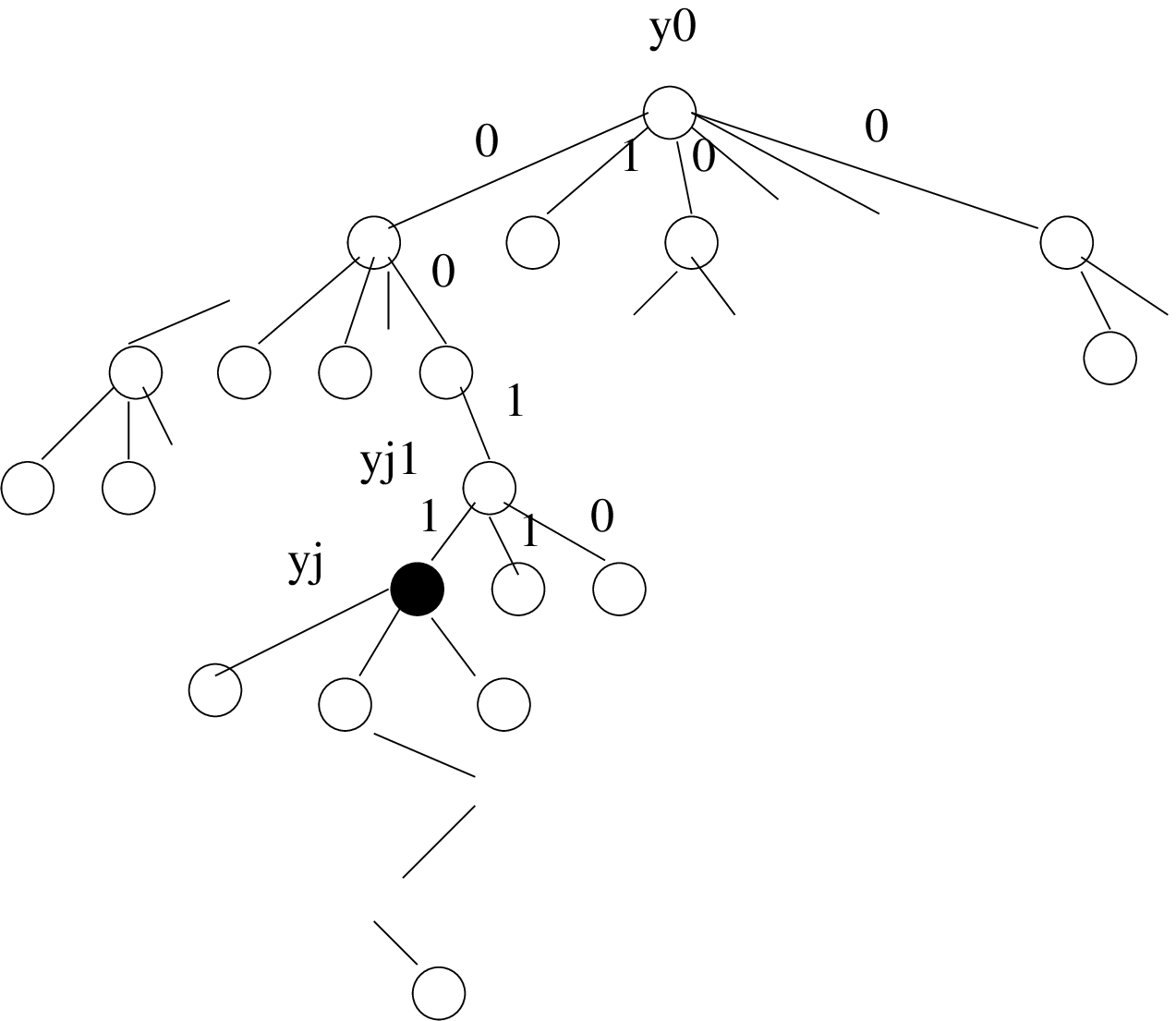}
\label{qsortfig}
\caption{Tree corresponding to the Backward Sequence Martingale corresponding to the 
comparisons for a fixed element
$x$. The root corresponds to $Y_0$ which denotes the 
expected running time. Every edge has cost 0 or 1 depending on whether $x$ and $X_i$ belong to
the same interval and a path in this tree reveals the indicator variables $I_j$. }
\end{figure}

So $\sum_{j=1}^{n} \E_{j-1}[ {( Y_j - Y_{j-1} )}^2 ] \leq \sum_{j=1}^{n}
 \frac{2}{n-j} \leq 2 H_n$where $H_n \leq \log n$. Plugging in $\lambda = 2 c \log n$ for 
some constant
$c$ and using Freedman's theorem, we obtain
\[ \Pr [ | Y_n - Y_0 | \geq c\log n ] \leq \exp \left(- \frac{ 4c^2 \log^2 n}
{2(\log n + c\log n /3)}     \right) \leq \frac{1}{n^c} .\] 
Note that $M_n = Y_i - Y_{i-1} \leq 1$.

This shows that a single element incurs at most $O(\log n)$ cost with high
probability and therefore quicksort runs in $O(n\log n)$ time with high 
probability. 

{\bf Remark} A straightforward application of the classic Azuma-Hoeffding bound \cite{MR95} 
\[ \Pr [ | Y_n - Y_0 | \geq t ] \leq \exp \left( \frac{- t^2}{\sum_{i=1}^n c_i^2 }
\right) \]
would not have been effective since the the bound $c_i = M_n =1$ makes the 
denominator too large for an $O(\log n)$ deviation bound. In \cite{Seidel93}, 
the author obtained a similar bound by using Chernoff bounds for binomial 
distribution by assuming
{\it independence} between $I_j$s across different levels. This assumption is 
an oversimplification since choices of pivots across different levels could
affect $I_j$'s. \\
 Also note that, there exists a superior bound of $O( n^{- \Omega 
(\log\log n)})$ for Quicksort obtained in
\cite{MH96}.

The above argument can be directly extended to obtain a concentration bound on 
the {\it dart throwing} game that has many applications (Mulmuley \cite{Mu88}). Consider throwing
$n$ darts randomly in $n$ ordered locations, say numbered $\{ 1, 2 \ldots n \}$.
Let $S(i)$ be a random variable that denotes the {\it smallest } numbered 
location among the first $i$ randomly thrown darts. Let $Z(i) = 1$ if $S(i) \neq
S(i-1)$ and $Z(1) = 1$. So $Z(i)$ is the number of times $S(i)$ changes among
the first $i$ darts thrown. We are interested in $\E [ Z(n) ]$ which can be
shown to be $\sum_{i=1}^{n} \frac{1}{i} = H_n$, the $n$-th harmonic. This 
follows from {\it backward analysis} by observing that among a set of $i$ 
randomly chosen numbers, the probability of picking the smallest number as the
last number is $\frac{1}{i}$. This is
related to many visibility problems in geometry as well as the analysis of 
Trieps. Using the Freedman's inequality, we can easily show the following 
from the previous argument and looking at the changes in the leftmost interval
induced by the darts.
\begin{coro}
\[ \Pr [  | Z(n) - H_n | \geq 0.9 \log n ] \leq \exp ( - 0.7 \log n ) \leq 
\frac{1}{n^{0.7}} \]
\label{dart}
\end{coro}
This implies that $\Pr [  0.1\log n \leq Z(n) \leq 1.9\log n ] \geq 1 - n^{-0.7}$.
The above result has been stated in a slightly weaker manner so that we can
claim a lower bound on $Z(n)$ that 
will be invoked later to show the limitations of \ric.

The analysis in this section also extends to problems like constructing
trapezoidal maps that can be used for point location (Seidel \cite{Seidel91}). 
Since
a trapezoid can be defined by at most 4 segments, the expected work for
point location is $\sum_{i=1}^{n} \frac{4}{j} \leq 4\log n $. Using a 
straightforward extension of the previous arguments, the following result can be
obtained.
\begin{lemma}
Given a set of $n$ non-intersecting line segments, a trapezoidal map 
can be constructed using \ric such that for any query point
$q$, the number times the trapezoid containing $q$ changes can be bounded
by $O(\log n)$ with inverse polynomial probability.
\label{ploc}
\end{lemma}
  
This result will turn out to be very useful for some later results.
\subsection{Comparison with an earlier bound}

We briefly recall the framework of Mehlhorn, Sharir and Welzl \cite{MSW93} 
to model the general RIC algorithm. A rooted $(n ,r)$ tree $T$ 
is either a single node for $r=0$ or (for $r > 0$) the
tree has $n$ children which are recursively defined $(n-1, r-1)$ subtrees.
Each of the $n$ edges has an associated weight $d_i$ corresponding to the
$i$-th child and $\max_{i=1}^n d_i \leq d(n)$ and $\sum_i d_i \leq M(n)$.
The expected cost of a path in this recursively defined tree is 
$A = \sum_{i=1}^{n-1} \frac{M(n-i)}{(n-i)}$. One of the main results in the paper is the following
tail bound (Theorem 1 in \cite{MSW93}).
\[ \Pr ( X \geq B ) \leq {\left( \frac{e}{1 + B/A} \right)}^{B/d(n)} \mbox{ for all } B \geq 0 \]
Although this bound looks somewhat simpler to use, this is not directly comparable to 
Freedman's bound except for some special cases    
like Lemma \ref{ploc} and quicksort where the concentration results are similar. 
It may be noted that the authors \cite{MSW93} analyze 
the backward execution of the algorithm for these results. 
This bound becomes weaker if $d(n)$ is not a constant - for some of the later applications $d(n)$
may be larger than $A$ in the worst case. The authors improve the bound for the
 specific problem of building visibility maps of line segments by using the expected 
value of $M(n)$. However, there is no generalization given for other problems.   

\ignore{This is indeed the case with our quicksort 
analysis where we can have more weight 1 edges in some paths as illustrated
by the following scenario. Consider a path $\Pi_1$ 
in the tree where the first $\sqrt{n}$
elements are picked and another path $\Pi_2$ where every $\sqrt{n}$th element is 
picked. At the $\sqrt{n}$ level, each node has $n - \sqrt{n}$ children. In path
$\Pi_1$, each of the children has weight 1 (as every pivot will be compared with
$x_j$) whereas in $\Pi_2$, only $\sqrt{n}$ elements will have weight 1. Therefore
$M(n- \sqrt{n})$ will have a bound $n - \sqrt{n}$ because of $\Pi_1$ although many 
nodes will have significantly less weight.
}


\section{Incremental Delaunay Triangulation}
\label{sec:delaunay}

We will now consider somewhat more complex scenarios like construction of
Delaunay Triangulation and three dimensional convex hull (see Guibas Knuth
and Sharir \cite{GKS92}). Broadly speaking these algorithms have two distinct
components - 
\begin{quote}
(i) Updating the (partial) structure of the points inserted
thus far.\\
(ii) Updating the point-location data structure of the uninserted points. 
\end{quote}
For concreteness, we will address the problem of Delaunay Triangulation.
The analysis corresponding to updating the point location structure is similar
to the analysis of quicksort given above. For the update of structural 
complexity, it was shown in \cite{GKS92} that the expected cumulative structural
change can be bound by $O(n)$, whereas for the latter, the expected work
over all the $n$ (random) insertions sequence $O(n\log n)$. In several places,
the authors in \cite{GKS92} pose the problem of tail estimates as an important
open problem. We will
do a combined analysis since we are interested in obtaining tail estimates
on the work including all data structural updates. 

In the remaining part of the paper,
we will be alluding to the BSM framework and make use of 
Equation \ref{varbound} for deriving the tail estimates. To avoid
any confusion, we will use stage/level $k$ to refer to the forward algorithm
when $k$ objects have been added and do all calculations in this order.
Although the martingale has been defined for the backward execution, 
substituting $n-k$ by $k$, consistently will not not affect anything except
the order of the summations. This will also help us use the random sampling
bounds without having to restate them in the flipped order. 
 
We will make use of the following result of \cite{CS89,HW86}.

\begin{theorem}
At any stage $i$ of the RIC of Delaunay triangulation, the $i$ randomly
chosen points is a uniform random subset of the $n$ points. So the number 
of unsampled points within each triangle is bounded by $O(\frac{n}{i} \log n)$
with probability $1 - 1/ n^c$ for any constant $c > 1$.
Moreover, all the triangles that {\it emerges} in the course of edge flips
also satisfy the above bounds.  
\label{epsnet}
\end{theorem}

{\bf Remark}: All the triangle that show up in the course of edge flips belong to
$\Pi^0 (R)$. Although some of them are not delaunay triangles and therefore, only
temporary, they can contribute to the running time, depending on if one maintains the
intermediate partitions.

To apply Freedman's bound, we will first bound the variance. Unlike the analysis
of quicksort, we will consider the work done for all the $n$ points
(actually $n-i$ uninserted points in stage $i$) together. Each edge flip
involves four triangles - two old and two new and redistributes the points
in the two new triangles. Since each triangle contains $O(\frac{n}{i}\log n)$ 
points w.h.p, each edge flip can be be done in $O(\frac{n}{i} \log n)$ w.h.p.
Since the maximum degree of a Delaunay graph of $i$ points is $i$, the total
number of edge flips in the $i$-th stage is bounded by $i$. Therefore we
can claim
\begin{lemma}
The work in stage $i$ of the algorithm, $i \leq n$ can be bounded by 
$O(n \log n)$ w.h.p.
\label{workbound}
\end{lemma} 

Let $\Pi_s ( R)$ denote the configurations in $\Pi^0 (R \cup s)$ 
adjacent to $s$ (or defined by $s$). The following
claims can be easily derived from some general random-sampling lemmas in
\cite{CS89}
\begin{lemma}
\[ \E[ \sum_{\sigma \in \Pi_s (R) } \ell (\sigma )] = O( \frac{n}{r})\E[ 
|\Pi_s (R)|] \]
\[ \E[ \sum_{\sigma \in \Pi_s (R) } \ell^2 (\sigma )] = O( \frac{n^2}{r^2}) 
\E[ |\Pi_s R) |] \]
\end{lemma}

{\bf Bounding Variance}

We will need the following result
\begin{lemma}
For real numbers $x_i \ \ 1 \leq i \leq r$
\[ { \left( \sum_{i=1}^r x_i \right)}^2 \leq r \left( \sum_{i=1}^r x_i^2 
\right) \] 
\end{lemma}
\begin{proof}
Using the convexity of the square function, from Jensens inequality it 
follows that 
\[ \frac{\sum_{i=1}^r x_i^2 }{r} \geq
{\left( \frac{ \sum_{i=1}^r x_i }{r} \right)}^2 \]
Multiplying both sides by $r^2$ yields the required result.
\end{proof}

The work done when a degree $j$ vertex $v$ is picked is proportional to the
number of points in the triangles adjoining the vertex. If $l(\sigma)$ is
the number of points in a triangle $\sigma$, then the work is proportional
to $T_k = \sum_{\sigma \in \Delta(v)} l (\sigma )$ 
where $\Delta(v)$ denotes
triangles adjacent to $v$. Squaring $T_k$ and taking expectation  
\begin{eqnarray*}
 \E [ T^2_k ] & = &  
\frac{1}{k} \sum_{v \in R^k} \E [ {( \sum_{\sigma \in \Delta (v)}
l (\sigma) ) }^2 ] \\
& \leq & \frac{1}{k} \sum_{v \in R^k} \E [ 
\sum_{v \in R^k } | \Delta (v) |  l^2 (\sigma)] \mbox{ from previous lemma }\\
 &  = & \frac{1}{k} \sum_{v \in R^k} \frac{n^2}{k^2} {| \Delta (v) |}^2 
\\
& = & O( \frac{ n^2}{ k } ) \ \mbox{ as } \sum_v {| \Delta (v) |}^2 = O( k^2 )
\end{eqnarray*}
This yields 
\[ W_n \leq \sum_{k=1}^{k=n} \E[ T^2_k ] 
 \leq \sum_{k=1}^{k=n} O( \frac{n^2 }{k}) = O( n^2 \log n) \]  

Plugging the bound of $M_n = O(n\log n)$ from Lemma \ref{workbound} in Freedman's theorem,
we obtain the following bound. 
\begin{lemma}
Let $T(n)$ denote the running time of ric based construction of Delaunay 
Triangulation and let $\lambda = c n\log n$ for a suitable constant $c$. Then
\[ \Pr [ T(n) \geq \alpha (n) \lambda ] \leq \exp \left( -
\frac{{(\alpha (n) cn\log n )}^2}{2 (n^2\log n + \alpha (n)  c \cdot n^2 \log n /3)}
\right) \leq \exp (-\alpha (n) ) \]
\label{genMbnd}
\end{lemma}
{\bf Remark} The above Lemma gives high 
probability bound for $T(n)$ exceeding $\Omega ( n\log^2 n)$ for $\alpha  = \Omega (
\log n)$.  However, this bound is 
superior to the straightforward Markov's
bound applied on the expected work as well as preferable to 
restarting the original algorithm using independent random bits each time.
There are also inputs (see Figure \ref{dtfig}) for which the inverse polynomial 
high probability bounds do not seem possible without changing the algorithm.


This analysis can be extended to the
three-dimensional convex
hull algorithm presented in Mulmuley \cite{Mu94}.
For fixed dimensional linear programming Seidel \cite{Seidel91a} proved a similar 
property and this can be extended to \ric algorithms like 
 closest pair \cite{KM:95}. However, these bounds are not very interesting and the 
reason could be viewed as follows. In the steps that the \ric algorithm for 
LP and closest-pair re-builds the data structure, there is a high cost that gets
subsumed in the expected bounds. But in Freedman's inequality, 
it is easily seen
that when $M_n = \Omega(\lambda (n))$, we are unable to obtain a inverse polynomial 
concentration bounds around $\lambda (n)$.    

\section{More generalized RIC : segment intersections}
\label{sec:segment}

We now consider a more general scenario in RIC (Randomized Incremental 
Construction). Using a {\it conflict graph} update model of RIC, we obtain
the following expression for expected work.

Expected work (\#edges created in  the conflict graph)=
\[ \sum_{\sigma \in \Pi^0 (R \cup s)} l (\sigma) \cdot \Pr\{ \sigma \in
\Pi^0 (R \cup s ) - \Pi^0 (R)  \} \]
From {\em backward analysis} this probability is the same as deleting
a random element from $R \cup s$ which is $\frac{d(\sigma )}{r+1}$.
By substituting this we obtain
\[ \sum_{\sigma \in \Pi^0 (R \cup s)} l (\sigma) \cdot \frac{d(\sigma )}{r+1}
 = \frac{d(\sigma )}{r+1} \sum_{\sigma \in \Pi^0 (R \cup s)} l (\sigma) \]
\[= O( \frac{d(\sigma )}{r} \cdot \frac{n}{r} \E [ \Pi^0 ( R \cup s )] ) \]
Therefore the expected work over the sequence of random insertions is
\[ \sum_{r=1}^{n} O( \frac{d(\sigma )}{r} \cdot \frac{n}{r} \E [ \Pi^0 ( R \cup s )] ) .\]

For the case of line segment intersections, it can be shown that 
$\E [ \Pi^0 ( R \cup s )] = O( r + \frac{ m \cdot r^2}{n^2} )$ from which
it follows that the expected work is 
\[ \sum_{r=1}^{n} O( \frac{d(\sigma )}{r} \cdot \frac{n}{r} \cdot 
O( r + \frac{ m \cdot r^2}{n^2} ) = \sum_{r=1}^{n} \left( \frac{dn}{r} + \frac{dm}{n} 
\right) = O(n\log n + m) . \]
Here $d (\sigma ) \leq 6$ which the maximum number of segments that define
a $\sigma$ (trapezoid in this case). 

Tail bounds for this problem has been elusive despite significant effort 
(see \cite{MSW93}). We will show that our previous techniques can be extended to
obtain tail estimates on the work done. 

Consider an arrangement of $n$ segments with $m$ intersections ($ 0 \leq m \leq
{n \choose 2}$). In the trapezoidal map ${\cal T}$ of the $n$ segments (also
known as a vertical visibility diagram), let us denote the set of trapezoids 
adjacent to segment $s_i$ by $T_i$. Any trapezoid $\sigma \in {\cal T}$ is 
defined by at most six segments. Since it is a planar map, and there are at most
$2n + 2m$ vertices, it follows that $\sum_i | T_i | = O(n + m)$. We would like to
obtain a bound on $\sum_i {| T_i |}^2$. Let us denote by $n_i$ and $m_i$ 
respectively, the number of segments end-points and intersection points visible
to segment $s_i$. It follows that $ | T_i | = O( n_i + m_i )$ and
\[ \sum_i {| T_i |}^2 = O( \sum_i {( n_i + m_i )}^2 )  = O( \sum_i n_i^2 + 
\sum_i n_i \cdot m_i + \sum_i m_i^2 ) \]  
where $m_i \leq O(n)$ from the zone theorem bound. Moreover $\sum_i n_i = O(n)$ 
and $\sum_i m_i = O(m)$ as each point is visible from the closest segments 
above and below. 

The first expression can be bounded by ${( \sum_i n_i )}^2 = O ( n^2 )$ and the
second expression by $2 ( \sum_i n_i ) \cdot ( \sum_i m_i ) = O( n \cdot m )$ 
(Cauchy-Schwartz inequality). The third expression is less than $ (m/n) \cdot
 n^2 = mn$. So, the overall expression can be bounded by $O( m\cdot n + n^2 )$.

For a uniformly chosen prefix $S^k$ of size $k$, the expected number of 
intersections in the sample is $ \frac{m k^2}{n^2}$, so the variance can be
bounded by
\[ \E[ T^2_k | S^k ] = \frac{1}{k} \E [ \sum_{s_i \in S^k}  {(\sum_{\sigma \in 
T_i } l (\sigma ) )}^2 ] \] 

To simplify calculations, we recall (Theorem \ref{epsnet} ) 
that $l(\sigma ) \leq 
O(\frac{n \log n}{k})$ with high probability. So, plugging this in the previous
expression, and using the previous bound on $\sum_i {| T_i |}^2$, 
we obtain (w.h.p.)
\[ \E[ T^2_k | S^k ] \leq \frac{1}{k} \cdot O(\frac{n^2\log^2 n}{k^2}) \cdot \E[ 
 k^2 + k \cdot m_k ] \]
where $m_k$ is the number of intersections in $S^k$. Taking expectation over
all choices of $S^k$, we obtain the unconditional expectation as
 \[ \E[ T^2_k ] \leq \frac{n^2 \log^2 n}{k^3} \cdot \E [ k^2 + k m_k ] \leq
   \frac{n^2 \log^2 n}{k^3} \cdot ( k^2 + \frac{m k^3}{n^2} )) \leq
 \frac{n^2 \log^2 n}{k} + m \log^2 n \]

This uses the bound $\E [ m_k ] = O( k + \frac{ m \cdot k^2}{n^2} )$.
This bound is relevant for the maintenance of conflict graphs. 

In contrast, for
an algorithm like Mulmuley \cite{Mu88}, where only the trapezoids are maintained,
the work done\footnote{there is some additional cost for point location that can
be bounded using Lemma \ref{ploc}}
can be bounded by using $l(\sigma ) = 1$ in the expression for $\E [ T_k^2 ]$.
This yields $\E[ T^2_k ] = O( k + m\frac{k^2}{n^2} )$. We shall return to this
case later. 

So 
\[ W_n \leq \sum_{k=1}^n \E[ T^2_k ] \leq O( n^2 \log^3 n + mn\log^2 n ) \]

To obtain high probability
bounds using Freedman's inequality, we want to bound this expression
by $\frac{ \lambda^2}{\log n}$ where $\lambda = c ( n\log n + m)$.
$M_n$ can be bounded by $O(n\log n \cdot \alpha (n)) $ and 
$M_n \cdot \lambda$ by $m n \alpha (n)$ with high probability.
This follows from a bound of $O(t \alpha(t))$ on the zone of a segment that 
intersects $t$ segments
in an arrangement of $n$ segments (\cite{Mu88}). 

So $\frac{\lambda^2}{W_n + M_n \cdot \lambda}$ can be bounded by 
\[ \frac{\Omega(m^2 + 
mn \log n + n^2 \log^2 n)}
{O( m n\log^2 n  + n^2 \log^3 n
+ m n\alpha (n)\log n + n^2 \alpha(n) \log^2 n } = 
 \frac{\Omega(m^2 + 
mn \log n + n^2 \log^2 n)}
{O( m n\log^2 n + n^2 \log^3 n }.\]

 So from Freedman's inequality we 
obtain a tail bound of $\exp( - \frac{m^2}{m n\log^2 n} )$ for 
$m \geq n \log^2 n$. 
\begin{theorem}
Let $T(n)$ represent the work done in the conflict-graph based 
segment intersection algorithm,
then there exists constant $\beta$,
such that for $m \geq \beta n \log^2 n$,
\[ \Pr [ T (n) \geq m ] \leq \exp ( - \frac{m}{n\log^2 n} ) \]
\end{theorem}
To the best of our knowledge, no prior concentration bound was known for the
conflict-graph based approach for segment intersection given by
Clarkson and Shor \cite{CS89}. The paper by \cite{MSW93} noted that their
methods could not be extended to this algorithm.

For the specific case of $m=0$, the bound can be improved 
by observing that the zone of a segment can be at most $O(n)$
(instead of $n\alpha(n)$) as there are no intersections. Setting $m=0$ in the
previous bound for $W_n$, we obtain the following
\begin{coro}
For constructing the trapezoidal map of $n$ non-intersecting line segments
using \ric , the work done $T(n)$ satisfies
\[ \Pr [ T (n) \geq c\beta n\log^2 n ] \leq n^{ - \beta^2  } \]
for some constant $c >1$.
\label{trapezoidmap}
\end{coro}
This shows that we can obtain inverse polynomial concentration bounds around
a running time that exceeds the expected running time by a factor of $\Omega (
\log n)$.

We now return to the algorithms of \cite{Mu88} and \cite{CS89} that do not 
maintain conflict-graphs but only involves segment end-points. As observed before,
the quantity $W_n$ can be bounded by 
\[ \sum_{k=1}^n O(k + m \frac{k^2}{n^2}) = O( n^2 + mn ) . \]
We summarize as follows.
\begin{lemma}
In the segment intersection algorithms of \cite{Mu88,CS89} 
that do not maintain conflict-graphs explicitly, 
the probability that the work exceeds $c (m + n\log n)$ can be bounded by
\[ \exp - \left( \frac{\Omega(m^2 +
mn \log n + n^2 \log^2 n)}
{O( m n\alpha(n) + n^2 \alpha (n) \log n) }\right) \leq \exp( - \frac{\log n}
{\alpha (n) }) \] 
since $M = O(n \alpha (n) )$.\\
 For $m \geq n\log n$, the bound improves to $\exp( - \frac{m}{n\alpha(n) })$.
\end{lemma}
{\bf Remark} This bound is better than the results
in \cite{MSW93} where the authors show that for some constant $\delta > 0$. 
\[ \Pr [ T(n) \geq Cm ] \leq 
\exp  \left( \frac{ - \delta m}{n \log n} \right) \mbox{ for }
m \geq n \log n \log\log\log n \]
\section{Can we improve the tail bounds}
\begin{figure}[t]
\psfrag{U}{{\small ${\cal U}$}}
\psfrag{B}{{\small ${\cal B}$}}
\psfrag{T}{{\small $T = c(n)\sqrt{n}$}}
\psfrag{n2}{{\small $\frac{n}{2}$ }}
\includegraphics[width=4.5in]{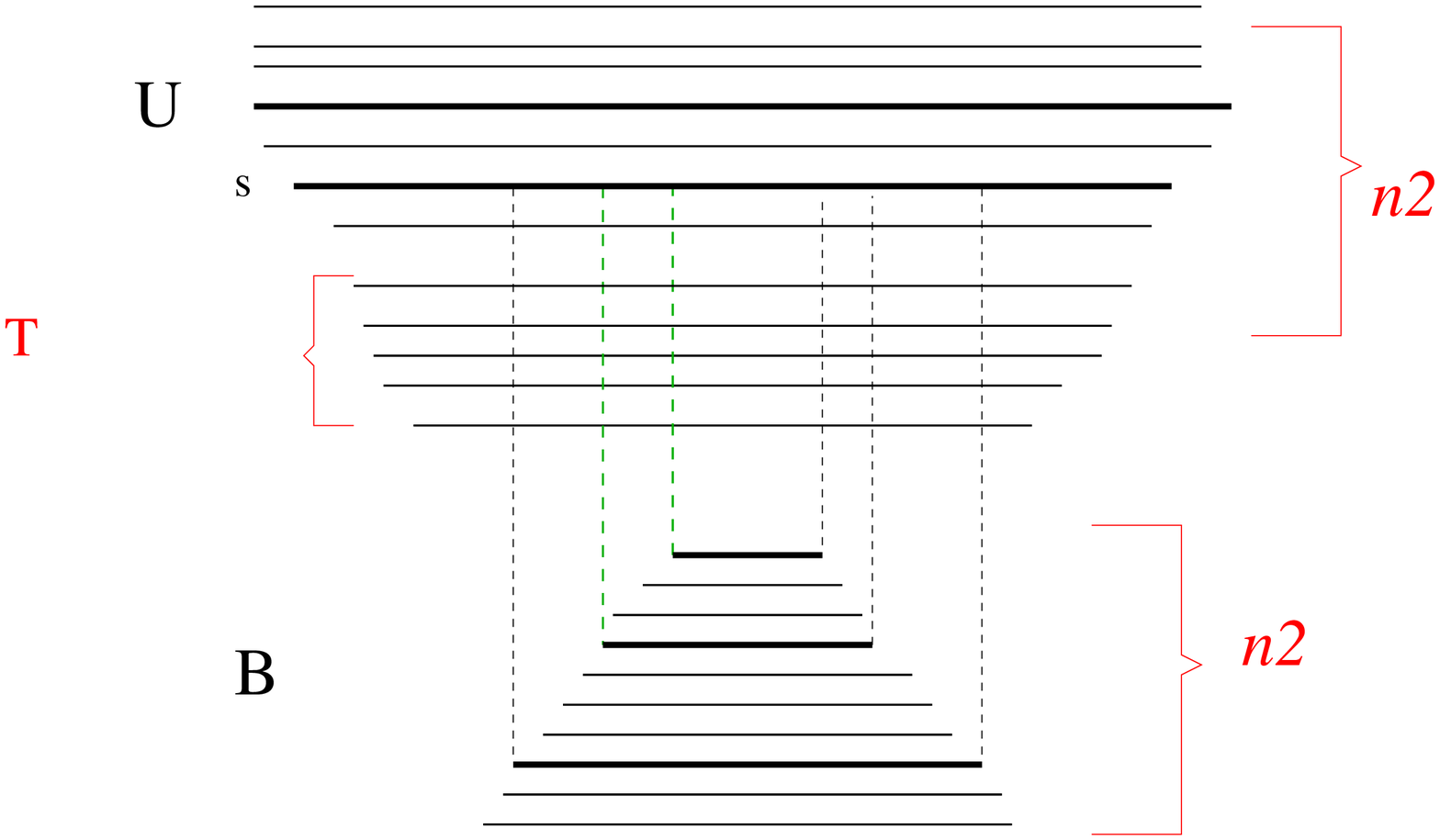}
\caption{A bad input for segments intersections (trapezoidal maps). The thicker
segments correspond to the sampled set.}
\label{trapezoid}
\end{figure}
Figure \ref{trapezoid} shows an input of $n$ horizontal segments divided into
two sets ${\cal U} , {\cal B}$ each of which has $n/2$ segments.
In the top pile ${\cal U}$ of the $n/2$ horizontal segments, let us consider 
the lowest $\sqrt{n} \cdot c(n)$ segments for some function $c(n)$ that will
determined in the analysis. 

We will consider the \ric after the first $10 \sqrt{n}$ insertions.
The number of segments in ${\cal B}$ and ${\cal U}$ 
respectively is at least $\sqrt{n}$ with high 
probability. 
Notice that, if the lowest sampled segment is $s \in {\cal U}$ then all the
segments below $s$ are {\it visible} to all
the end-points in the lower pile ${\cal B}$ of the segments. Namely, if there 
are $m$ unsampled segments below $s$, then the size of the conflict graph 
is at least 
$\sqrt{n} \cdot m$. Consider the set of $\sqrt{n} c(n)$ segments and denote this 
by $T$. 

Among the set of $10 \sqrt{n}$ segments, what is the probability that none of
them was sampled from the set $T$. This can be easily seen as
\[ {( 1 - \frac{c(n)}{\sqrt{n}} )}^{10 \sqrt{n}} = \Omega ( 4^{ - 10 c(n)})  \] 

Every time $s$ changes, new edges are created in the conflict graph by the
sampled edges in ${\cal B}$. Following the initial $10 \sqrt{n}$ segments, 
let us consider the second phase where segments from $T$ may be sampled.
Within $T$, let us denote by $T'$ the lowest $\alpha c(n) \sqrt{n} $ 
segments for some constant $\alpha < 1$ and let $T''$ denote the 
remaining segments. 
What is the probability
that among the first $\log n$ segments sampled from $T$, none are from $T'$ ?
This can be calculated as ${( 1 - \alpha )}^{\log n} = \Omega ( 4^{- \alpha
\log n})$.
    
It follows from {\it backward analysis} that the lowest sampled segment in 
$T''$ changes about $\theta (\log \log n) $ times 
(Corollary \ref{dart}).  
So, at least $\Omega ( \sqrt{n} c(n) \cdot \sqrt{n} \log\log n)$ edges are 
created with
probability $\Omega ( \exp ( - 20 c(n) )$. For example, for $c(n) = \log n/40$,
then the probability is at least $\frac{1}{\sqrt{n}}$ that the total
number of edges created in the conflict graph is 
$\Omega ( n\log n \log\log n )$.
More specifically, this holds with probability $\frac{1}{\log n} 
\cdot  \frac{1}{n^{\alpha}} \cdot \frac{1}{\sqrt{n}} \cdot
( 1 - \frac{1}{n^{\beta}})$ for some constants $0 < \alpha , \beta < 1$ which
is $\Omega ( \frac{1}{\sqrt{n}})$. The multiplicative factors (very close to 1)
help us uncondition the probability that 
\begin{quote}
(i) Adequate samples - about $\sqrt{n}$ are chosen from sets ${\cal U , B}$ and
\\
(ii) There are $\theta (\log\log n)$ changes in $s$. 
\end{quote}

The above argument can be easily generalized as follows
\begin{theorem}
There exists inputs for which the conflict-graph based 
\ric algorithm for constructing vertical 
visibility maps (segment intersections with no intersections) 
encounters $\Omega ( c(n) n\log n\log\log n  )$ structural changes 
with probability $\Omega ( e^{ -20 c(n) })$ for $c(n) = o( \sqrt{n})$. 
\end{theorem}
In particular, by choosing $c(n) = \frac{\log n}{40}$, the conflict 
graph based 
\ric algorithm may encounter $\Omega ( n\log\log n)$ changes 
with probability $\Omega ( \frac{1}{\sqrt{n}} )$ that
rules out inverse polynomial bounds for a total work of $O(n \log n )$.
Note that $\frac{1}{\sqrt{n}}$ can be easily increased to 
$\frac{1}{n^{\epsilon}}$ for any $\epsilon > 0$.

By comparing this result with Corollary \ref{trapezoidmap}, we have nearly tight
bounds for tail estimates, modulo some constant factor in the work done.

\bibliographystyle{plain}
\bibliography{ric}
\section{Appendix}

We provide a brief description of the notations and definitions that capture the
framework of RIC and its analysis in very general setting.

Given a set $S$ of $n$ elements (like points, segments, lines etc.), a configuration
$\sigma$ is defined by at most $d$ objects where $d$ is $O(1)$. 
The set of objects
is denoted by $d(\sigma )$ and 
the number
of configurations is bounded by $n^d$ if there are no more than $O(1)$ configurations
associated each subset of $d$ elements (there can be more than one configuration
associated with the same $d(\sigma )$ elements.

Let $\ell (\sigma )= S \cap \sigma - d (\sigma )$
be the elements that intersect with $\sigma $. With a slight overloading of
notation we will also use $\ell (\sigma )$ to denote the set of the intersecting 
elements with $\sigma $ also. Let $\Pi^i (S)$ denote the set of configurations $\sigma$
with $\ell (\sigma ) = i$. We use $\Pi (S) = \cup_i \Pi^i (S)$ to denote all 
configurations. For any subset $R \subset S$, we use $\Pi (R)$ to denote the 
configurations defined by elements of $R$ and the {\it conflict list} of any 
configuration $d(\sigma ) \subset R$ as $\sigma \cap S$, i.e., all the elements and
not just the elements in $R$.   

A {\it conflict graph} represents the relation between the 
configurations in $\Pi^0 (R)$ and the 
corresponding
conflict list, which is a bipartite graph with configurations in $\Pi^0 (R)$ on one
side and the uninserted elements on the other side. 
Randomized Incremental construction can be thought of as maintaining and update of
the conflict graph starting with $R = \phi$ and successively adding a random 
(uninserted) element $e \in S - R$ into $R$. This introduces $\sigma \in \Pi^{0}( R 
\cup e) - \Pi^0 (R)$ requiring appropriate changes in the conflict graph. 
\begin{figure}[t]
\psfrag{s1}{{\small $s_1 $}}
\psfrag{s2}{{\small $s_2 $}}
\psfrag{s3}{{\small $s_3 $}}
\psfrag{s4}{{\small $s_4 $}}
\psfrag{s5}{{\small $s_5 $}}
\psfrag{s6}{{\small $s_6 $}}
\psfrag{s7}{{\small $s_7 $}}
\psfrag{s8}{{\small $s_8 $}}
\psfrag{s9}{{\small $s_9 $}}
\psfrag{s10}{{\small $s_{10} $}}
\begin{center}
\includegraphics[width=4in]{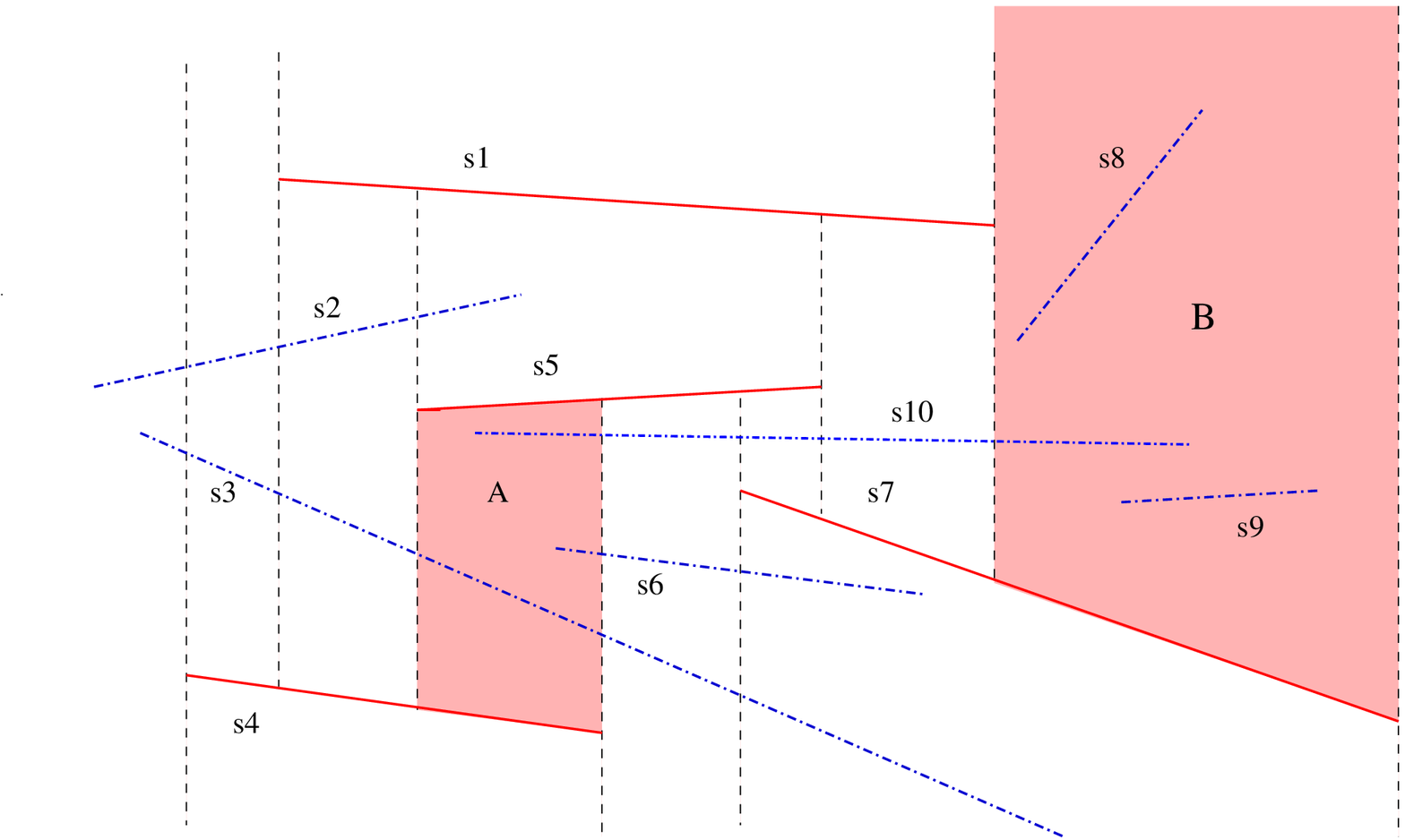}
\end{center}
\label{configtrap}
\caption{The red segments are sampled segments and blue segments are unsampled.
The trapezoids $A, B$ are configurations that belong to $\Pi^{0} (R)$. Here
$d (A) = \{ s_4 , s_5 \} \  \ell ( A) = \{ s_3 , s_6 , s_{10}\} $.} 
\end{figure}
To illustrate this framework on qiucksort, we define the configurations as intervals
defined by a pair of elements $[ x_i , x_j ]$ where $x_i < x_j$. Initially there 
is the hypothetical configuration $( - \infty , + \infty )$. As we introduce more 
pivots, we maintain the ordered set of intervals induced by the elements chosen as
pivots. As we introduce a pivot, some interval is split. Eventually we have the 
sorted set defined by consecutive intervals. When an interval $[ x_i , x_j ]$ splits
because of a pivot element $y$ such that $ x_i < y < x_j$, the elements in 
$\ell ( [ x_i , x_j ]) \cap S$ is reassigned to $\ell( [ x_i , y ])$ and $\ell ( [ y, 
x_j ])$ appropriately. The number of comparisons required is roughly 
$| \ell ( [ x_i , x_j ]) \cap S |$ (the cardinality). 

The analysis of quicksort in this framework 
can be done using the technique of {\it backward analysis} which is very elegant.
Let us assign an indicator random variable $X_k$ 
associated with
an element $x$, such
that
\[ X_k = \left\{ \begin{array}{ll}
        1 & \mbox{ if $x$ is compared for the $k$-th pivot} \\
       0 & \mbox{ otherwise }
      \end{array}
 \right. \]
The number of comparisons involving $x$ is given by $\sum_{k=1}^{n} X_k$. Therefore
\[ \E[ \sum_{k=1}^{n} X_k] = \sum_{k=1}^n \E[ X_k ] = \sum_{k=1}^{n} p_k (x) \]
where $p_k (x)$ is the probability that element $x$ is
involved in the partitioning of the $k$-th pivot insertion. 

To compute the probability, we observe that $X_k =1$ iff the $k$-th pivot $y$ 
is one of the two elements that bound the interval containing $x$ after $k$ pivots
are chosen randomly. For a fixed choice of $k$ initial pivots, the probability that 
$y$ is one of the two bounding elements is at most $\frac{2 (k-1)!}{k!} = 
\frac{2}{k}$. The numerator represents the number of permutations with one of the
bounding elements being the last pivot. Although this is the probability conditioned
on the choice of the first $k$ pivots, clearly unconditioning would also give us the
same probability. Therefore the expected number of comparisons involving $x$ is
$\sum_{k=1}{n} \frac{1}{k} = O(\log n)$. Further the total expected 
number of comparisons is $O(n\log n)$  by summing over all elements.

\end{document}